\documentclass{article}
\usepackage{graphicx} % Required for inserting images
\usepackage{tikz}
\usetikzlibrary{arrows.meta,positioning}

\usepackage[dvipsnames]{xcolor}
\usepackage[letterpaper, left=1in, right=1in, top=0.9in, bottom=0.9in]{geometry}
\usepackage[american]{babel}
\usepackage[normalem]{ulem}
\usepackage{amsmath, amssymb, cases, amsthm}
\usepackage{mdframed,lipsum}
\usepackage{thmtools, thm-restate}
\usepackage[shortlabels]{enumitem}
\usepackage{mdframed}
\usepackage{bbm}
\usepackage{bm}
\usepackage{microtype}
\usepackage{xcolor}
\usepackage{makecell}
\usepackage{mathtools}
\usepackage{algorithmic}
\usepackage[procnumbered,ruled,vlined,linesnumbered]{algorithm2e}
\usepackage{float}
\usepackage{varwidth}
\usepackage{svg}
\usepackage{tikz-network}
\usetikzlibrary{decorations.pathreplacing}
\usepackage{subcaption} % For subfigures and labeling (a), (b)
\usepackage{multirow}
\usepackage{natbib} 
\usepackage{tablefootnote}
\usepackage{enumitem}

\SetKwInput{KwData}{Input}
\SetKwInput{KwResult}{Output}
\SetKwInput{KwGlobalVar}{Global variables}
\SetKw{Break}{break}

\usepackage[bookmarks,colorlinks,breaklinks]{hyperref}
\hypersetup{urlcolor=blue, colorlinks=true, citecolor=green!50!black, linkcolor=black}
\usepackage[capitalize,nosort,nameinlink]{cleveref}

\crefname{algocfline}{line}{lines}

\usepackage{tikz-network}
\usepackage{authblk}

\DeclareMathOperator*{\argmax}{arg\,max}

\declaretheorem[numberwithin=section,refname={Theorem,Theorems},Refname={Theorem,Theorems}]{theorem}

\declaretheorem[numberlike=theorem]{lemma}

\declaretheorem[numberlike=theorem]{definition}

\declaretheorem[numberlike=theorem, refname={Observation,Observations},Refname={Observation,Observations},name={Observation}]{observation}

\theoremstyle{definition}
\newcommand{\pot}{$\phi(\X)$}

\newcommand{\poto}{$\phi_1(\X)$}
\newcommand{\pott}{$\phi_2(\X)$}

\newcommand{\xs}{$(\X,\sigma)$}

\newcommand{\aij}{A_{i,j}(\X,\sigma)}

\newcommand{\uj}{U_j(\X)}
\newcommand{\ui}{U_i(\X)}

\newcommand{\bijo}{B_{i, j}^{(1)}(\X)}
\newcommand{\bijt}{B_{i, j}^{(2)}(\X)}
\newcommand{\bjio}{B_{j, i}^{(1)}(\X)}
\newcommand{\bjit}{B_{j, i}^{(2)}(\X)}

\newcommand{\biko}{B_{i, k}^{(1)}(\X)}

\newcommand{\bkio}{B_{k, i}^{(1)}(\X)}

\newcommand{\bio}{B_{i}^{(1)}(\X)}
\newcommand{\bit}{B_{i}^{(2)}(\X)}
\newcommand{\biu}{B_{i}^{(u)}(\X)}

\newcommand{\bko}{B_{k}^{(1)} (\X)}

\newcommand{\bo}[1]{B_{#1}^{(1)}(\X)}
\newcommand{\bt}[1]{B_{#1}^{(2)}(\X)}

\newcommand{\boxy}[2]{B_{#1,#2}^{(1)}(\X)}
\newcommand{\btxy}[2]{B_{#1,#2}^{(2)}(\X)}

\newcommand{\efx}{\textsf{EFX}}
\newcommand{\ef}{\textsf{EF}}

\newcommand{\mms}{\textsf{MMS}}
\newcommand{\pmms}{\textsf{PMMS}}
\newcommand{\eefx}{\textsf{EEFX}}
\newcommand{\mxs}{\textsf{MXS}}
\newcommand{\efl}{\textsf{EFL}}
\newcommand{\efr}{\textsf{EFR}}
\newcommand{\rmms}{\textsf{RMMS}}
\newcommand{\effx}{\textsf{EF2X}}
\newcommand{\po}{\textsf{PO}}

\newcommand{\X}{\mathbf{X}}
\newcommand{\Y}{\mathbf{Y}}

\newcommand{\C}{\mathbf{C}}

\title{$\efx$ Allocations Exist on Triangle-Free Multi-Graphs}
\author[1]{Mahyar Afshinmehr \thanks{mahyarafshinmehr@gmail.com}}
\author[2]{Arash Ashuri\thanks{arashashoori199821@gmail.com}}
\author[3]{Pouria Mahmoudkhan \thanks{pouriamksh@gmail.com}}
\author[4]{Kurt Mehlhorn \thanks{mehlhorn@mpi-inf.mpg.de}}

\affil[3]{Department of Computer Engineering, Sharif University of Technology}
\affil[4]{Max Planck Institute for Informatics, Saarland Informatics Campus (SIC)}
\date{}

\begin{document}

\maketitle

\begin{abstract}
    We study the fair allocation of indivisible goods among agents, with a focus on limiting envy. A central open question in this area is the existence of \emph{$\efx$ allocations}---allocations in which any envy of any agent $i$ towards any agent $j$ vanishes upon the removal of any single good from $j$'s bundle. Establishing the existence of such allocations has proven notoriously difficult in general, but progress has been made for restricted valuation classes. \citet{CFKS23} proved existence for \emph{graphical valuations}, where goods correspond to edges in a graph, agents to nodes, and each agent values only incident edges. The graph was required to be simple, i.e., for any pair of agents, there could be at most one good that both agents value. The problem remained open, however, for \emph{multi-graph valuations}, where for a pair of agents several goods may have value to both. In this setting, \citet{SS25} established existence whenever the shortest cycle with non-parallel edges has length at least six, while \citet{ADKMR25} proved existence when the graph contains no odd cycles.  
    
    In this paper, we strengthen these results by proving that $\efx$ allocations always exist in multi-graphs that contain no cycle of length three. Assuming monotone valuations, we further provide a pseudo-polynomial time algorithm for computing such an allocation, which runs in polynomial time when agents have cancelable valuations, a strict superclass of additive valuation functions. Accordingly, our results stand as one of the only cases where $\efx$ allocations exist for an arbitrary number of agents.
\end{abstract}

\section{Introduction}

    Research on \emph{Fair Division} seeks to formally understand how a collection of resources can be distributed among multiple agents in a manner that is deemed \emph{fair} by all participants. The problem arises across a wide range of domains, including the division of business assets, allocation of computational resources, course assignments, and dispute resolution settings such as divorce settlements and air traffic management (\cite{etkin2007spectrum,moulin2004fair,vossen2002fair,budish2012multi,pratt1990fair}). Originating with the foundational work of \citeauthor{S49}, fair division has developed at the intersection of economics, mathematics, computer science, and social choice theory. In recent years, it has grown into a mature and interdisciplinary area of study, motivated by both its theoretical depth and its practical significance. For detailed surveys and expository treatments, see \cite{AABRLMVW23,brams1996fair,brandt2016handbook,robertson1998cake}.

    Over the past decade, research on fair division has primarily focused on allocating indivisible goods, with the central goal of minimizing envy among agents. Since some degree of envy is inevitable in certain instances\footnote{For example, if a single good is more valuable than all others combined, any agent who receives it will inevitably be envied.}, much of the literature has concentrated on relaxations of exact envy-freeness and on understanding the extent to which these relaxations can be guaranteed or approximated.

    The first relaxation of envy-freeness is {\em envy freeness up to some good} ($\ef1$), introduced formally in \cite{B10} and shown achievable earlier by 
    \cite{LMMS04}.

    An allocation is $\ef1$ if for any two agents $i,j$ with assigned bundles $X, Y$, respectively, 
    there exists some good $g \in Y$ such that $i$ weakly prefers $X$ to $Y\setminus\{g\}$.

    A particularly compelling notion is \emph{envy-freeness up to any good ($\efx$)}, introduced by \citet{CKMPSW19}. An allocation is $\efx$ if removing any good from the bundle allocated to an agent would ensure that no agent envies the remaining bundle. The $\efx$ criterion strikes a delicate balance between fairness and feasibility, and has therefore emerged as a central benchmark in the study of indivisible goods allocation.
     
    Despite its appeal, the existence of $\efx$ allocations in general remains a major open problem. Over the last few years, significant progress has been made by identifying valuation classes where $\efx$ allocations are guaranteed to exist.
    
   Today, we know that $\efx$ exists in restricted settings, such as when  agents' valuations are identical \cite{PR20}, 
    monotone but only up to two types \cite{M24},
    additive but only up to three types \cite{PGNV25}, or binary \cite{BSY23}. 
    For three-agent instances, $\efx$ allocations were shown to exist for additive valuations \cite{CGM24}, subsequently extended to nice cancelable valuations \cite{BCFF21} and $\mms$-feasible valuations \cite{ACGMM23}. 
    
    \citet{CFKS23} showed that $\efx$ allocations exist when valuations can be represented via a simple graph, where goods correspond to edges in a graph, agents correspond to nodes, and each agent values only the edges incident to her. 
    A natural application of this setting arises in the allocation of geographical resources—for example, distributing natural resources among neighboring countries, assigning office spaces to research groups, or dividing public areas among communities within a region.
    
    The natural next step is to consider \emph{multi-graph valuations}, an open question raised by \cite{CFKS23}, where multiple goods may connect the same pair of agents. 
    This setting is substantially richer and more complex. Understanding whether $\efx$ allocations always exist in this model has emerged as a key open question.
    Independently, \citet{ADKMR25, SS25, BP24} established the existence of $\efx$ allocations when the underlying multi-graph contains no odd cycles. 
    Moreover, \cite{SS25} proved that $\efx$ allocations exist whenever the shortest cycle with non-parallel edges has length at least six, 
    while \cite{BP24} further showed that if the underlying graph is t-colorable and its shortest cycle has length at least $2t-1$, then an $\efx$ allocation is guaranteed. \cite{ADKMR25} has also proved the existence of $\efx$ allocations on multi-graphs when the underlying graph is a single cycle.

\subsection{Our contributions and Techniques}
    In this work, we answer the open question raised by \cite{CFKS23} and strengthen the previously known results in the multi-graph setting and take a significant step toward resolving the general existence question for multi-graph valuations. Specifically, we prove that $\efx$ allocations always exist when the underlying multi-graph contains no cycle of length three. This substantially improves upon the previous six-cycle bound of \cite{SS25}, and pushes further the no odd-cycle result of \cite{ADKMR25}.    
    We present a pseudo-polynomial time algorithm for computing an $\efx$ allocation under monotone valuations. Furthermore, when agents have cancelable valuations, a strict superclass of additive valuations, our algorithm runs in polynomial time (see Theorem \ref{thm:main_result}). 

    In general, our results go strictly beyond the work of \citet{ADKMR25, SS25, BP24} and expand the frontier of known valuation classes for which $\efx$ allocations are guaranteed to exist, and provide efficient methods for computing them.

\subsubsection*{Technical Overview}

    We give a brief and simplified description of our techniques to prove the existence of $\efx$ allocations on triangle-free multi-graphs. Our Algorithm consists of three phases that move in the space of partial $\efx$ allocations.

\paragraph{Phase One.}
    For each pair of agents $i$ and $j$, we partition the set of edges between them into two bundles, which we refer to as unit bundles. More precisely, we designate an agent to divide the set of items shared by $i$ and $j$ into two $\efx$-feasible\footnote{Given a partition $\Y = \langle Y_1, \cdots, Y_l \rangle$ of a set $S$, we say a bundle $Y_k$ is $\efx$-feasible for an agent like $i$, if $i$ weakly prefers $Y_k$ to every other bundle in the partition after the removal of any item from it.} bundles from her perspective. These unit bundles remain fixed throughout the algorithm: if an agent receives one item from a unit bundle, she receives the entire bundle.

    In Phase One, each agent $i$ is assigned exactly one of her incident unit bundles, subject to the following conditions:
\begin{enumerate}
    \item The resulting partial allocation is $\efx$.
    \item Each agent weakly prefers her own bundle to every unallocated incident unit bundle.
    \item The longest path in the envy graph\footnote{The envy graph associated with a partial allocation has one node per agent, and an edge from agent $i$ to agent $j$ whenever $i$ envies $j$’s bundle relative to her own.} has length at most one.
\end{enumerate}

    To ensure these properties, agents select their most preferred unallocated incident unit bundle in a carefully chosen order. Both the selection order and the construction of unit bundles are essential for maintaining the above invariants. For instance, if agent $i$ comes before agent $j$ in this order, then agent $j$ cuts the set $E(i, j)$\footnote{$E(i, j)$ denotes the set of items (edges) between agents $i$ and $j$.} into two $\efx$-feasible bundles for herself, that are the unit bundles between agents $i$ and $j$.

    If, at the end of Phase One, no agent envies another, then the remaining unit bundles can be allocated as follows: for each pair of agents $i$ and $j$, one of the unit bundles between them is assigned to $i$, and the other to $j$, while ensuring that each agent retains the unit bundle assigned to her in Phase One. This yields a complete $\efx$ allocation. The justification is that, for every pair $(i,j)$, the value of $j$’s bundle from $i$’s perspective equals the value of the unit bundle incident to both $i$ and $j$ that is allocated to $j$. By construction, this value does not exceed the value of $i$’s own bundle—since that unit bundle is either (a) unallocated after Phase One, in which case $i$ does not envy it, or (b) allocated to $j$, in which case $i$ does not envy it by assumption.

    If, however, some agents are still envied after Phase One, then every non-envied agent can nevertheless be assigned one of her incident unit bundles while preserving the $\efx$ property, by an analogous argument. The remaining challenge lies in allocating the unassigned unit bundles incident to envied agents.

\paragraph{Phases Two and Three.}
    The second and third phases of the algorithm are designed to allocate these remaining unit bundles that are incident to envied agents. After phase one, the envy graph would look like a set of stars (see \cref{fig 1}) since the length of its longest path is at most one.

\begin{center}
\begin{figure}[H]
    \centering

\begin{tikzpicture}[scale=0.7, every node/.style={circle,fill=black,inner sep=1pt}] 
% --- Star 1 (3 leaves) ---
\node (A) at (0,0) {};
\foreach \i in {0,120,240}{
  \node (A\i) at (\i:1.2) {};
  \draw[->] (A) -- (A\i);
}

% --- Star 2 (4 leaves) ---
\node (B) at (4,0) {};
\foreach \i in {45,135,225,315}{
  \node (B\i) at ($(B)+(\i:1.2)$) {};
  \draw[->] (B) -- (B\i);
}

% --- Star 3 (5 leaves) ---
\node (C) at (8,0) {};
\foreach \i in {90,162,234,306,18}{
  \node (C\i) at ($(C)+(\i:1.2)$) {};
  \draw[->] (C) -- (C\i);
}
\node (D) at (-4,0.6) {};
\node (E) at (-4,-0.6) {};
\end{tikzpicture}
    \caption{This is an example of how the envy graph may look after phase one.}
    \label{fig 1}
\end{figure}
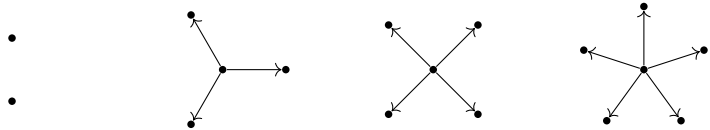
\end{center}

    Consider an arbitrary star centered at a non-envied agent $i_0$, who envies agents ${i_1, i_2}$, and also consider some additional agent $j$ (see \cref{fig: 2}). Let $(C_1, C_2)$ denote the unit bundles between agents $i_1$ and $j$. Similarly, let $(D_1, D_2)$ denote the unit bundles between agents $i_2$ and $j$. We attempt to allocate the remaining incident unit bundles $C_1$ and $D_1$ to agent $i_0$, allowing the allocation to assign an item to a non-incident agent.

    Note that agent $j$ does not envy $i_0$’s original bundle $X_{i_0}$, since the longest path in the envy graph has length one. Therefore, if $j$ does not envy the combined bundle $C_1 \cup D_1$, then she also does not envy $C_1 \cup D_1 \cup X_{i_0}$. This follows from the triangle-free property of the multi-graph, which ensures that at least one of the sets $C_1 \cup D_1$ or $X_{i_0}\cap E(i_0,j)$ is empty. 

    If the attempt of allocating $C_1$ and $D_1$ to agent $i_0$ fails to maintain an $\efx$ partial allocation, we adjust the allocation greedily. For instance, if $C_1 \cup D_1$ is more valuable than agent $j$’s current bundle, we instead allocate $C_1 \cup D_1$ to $j$. Furthermore, after adding these unit bundles for every possible $j$, if agent $i_1$ comes to envy $i_0$, we swap the bundles of $i_0$ and $i_1$, thereby making both agents happier.

    The second and third phases of our algorithm consist of a sequence of such local adjustments—allocations and unallocations of unit bundles—designed to preserve the $\efx$ property while ensuring that all items are eventually allocated.

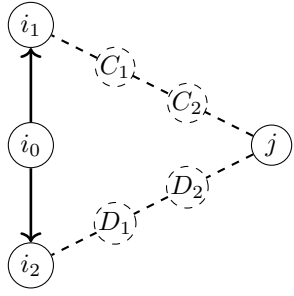
\begin{figure}[h]
    \centering
\begin{center}
\begin{tikzpicture}[scale=0.8, every node/.style={circle,fill=black,inner sep=1.5pt}]
% --- Star center  i0 ---
\node[circle,draw,fill=white,minimum size=6mm,inner sep=0pt] (i0) at (0,0) {$i_0$};

% --- 2 envied nodes ---
\node[circle,draw,fill=white,minimum size=6mm,inner sep=0pt] (i1) at (90:2.0) {$i_1$};
\node[circle,draw,fill=white,minimum size=6mm,inner sep=0pt] (i2) at (270:2.0) {$i_2$};

% Directed edges (i0 -> envied nodes)
\draw[->, line width=1pt] (i0) -- (i1);
\draw[->, line width=1pt] (i0) -- (i2);

% --- External node j ---
\node[circle,draw=black,fill=white,inner sep=2pt] (j) at (4,0) {$j$};

% --- Dashed edges with dashed circles ---
% Dashed line to i1
\draw[dashed, line width=0.8pt] (j) -- (i1);
\node[circle,draw,dashed,fill=white,minimum size=2.5mm,inner sep=0.2pt] at ($(j)!0.35!(i1)$) {$C_2$};
\node[circle,draw,dashed,fill=white,minimum size=2.5mm,inner sep=0.2pt] at ($(j)!0.65!(i1)$) {$C_1$};

% Dashed line to i2
\draw[dashed, line width=0.8pt] (j) -- (i2);
\node[circle,draw,dashed,fill=white,minimum size=2.5mm,inner sep=0.2pt] at ($(j)!0.35!(i2)$) {$D_2$};
\node[circle,draw,dashed,fill=white,minimum size=2.5mm,inner sep=0.2pt] at ($(j)!0.65!(i2)$) {$D_1$};

\end{tikzpicture}
\end{center}
        \caption{This figure shows a star in the envy graph, and an arbitrary agent $j$, and it depicts unit bundles among envied agents of this star and agent $j$.}
    \label{fig: 2}
\end{figure}

\subsection{Further Related Works}

    In another line of research,
    \citet{ARS24} showed that every additive multi-graph instance admits a $\frac{2}{3}$-$\efx$ allocation, and \citet{KKSS25} 
    recently improved this bound to $\frac{1}{\sqrt{2}}$-$\efx$. 
    Furthermore, \citet{KKSS25} showed that $\efx$ allocations exists on multi-graphs if the valuation functions are restricted additive.
    \citet{ZWLL24} studied the mixed manna setting with both goods and chores, proving that deciding the existence of $\efx$ orientations on simple graphs with additive valuations is NP-complete. 
    \citet{ZM25} established a connection between the existence of $\efx$ orientations and the chromatic number of the graph.
    \cite{BGRS25} refined this complexity landscape by showing that bipartiteness is a tight boundary for tractability, identifying hardness even for graphs very close to bipartite.
    More recently, \citet{DEGK25} proved that $\ef1$ orientations always exist for monotone valuations and can be computed in pseudo-polynomial time.

    For more general settings, a lot of work has focused on relaxations of $\efx$.
    One such relaxation has aimed to achieve multiplicative approximations of $\efx$. \citet{PR20} showed the existence of $1/2$-$\efx$ allocations for subadditive valuations, and \citet{AMN20} proved the existence of $1/\phi\approx 0.618$-$\efx$ allocations for additive valuations. 
    $2/3$-$\efx$ allocations were shown to exist for more restrictive settings like
    when there are at most four types of valuations \citep{HMN25}, at most seven agents \citep{ARS24}, and when the agents agree on what are the top $n$ items \citep{MS23}. \citet{BKP24} achieved improved $\efx$ approximations for restricted settings. 

    A related relaxation, $\effx$, requires that envy disappear after the removal of any two goods from an envied bundle.
    $\effx$ allocations were shown to exist for four agents with cancelable valuations \citep{AGS25},
    any number of agents with restricted additive valuations \citep{ARS22}, and with $(\infty,1)$-bounded valuations \citep{KSS24}.
    Several other relaxations of $\efx$ have also been studied, such as $\efl$ \citep{BBMN18}, $\efr$ \citep{FHLSY21}, 
    $\eefx$ and $\mxs$ \citep{CGRSV22} (subsequently extended to monotone valuations by \cite{HANR25}).
    In addition, several works have explored the simultaneous satisfaction of multiple fairness notions. 
    For instance, \citet{AR25} established the coexistence of $2/3$-$\mms$ and $\ef1$, \citet{AG25} combined $\mxs$ with $\efl$, and \citet{F25} achieved $\rmms$ alongside $\efl$.

    Another line of research has focused on ``partial allocations,'' donating some of the goods to   charity and achieving $\efx$ with the rest. This was first studied by \citet{CGH19} who showed the existence of a partial allocation that satisfies $\efx$ and its Nash social welfare is half of the maximum possible. \citet{CKMS20} proved that $\efx$ allocations exist for $n$ agents if up to $n-1$ goods can be donated; moreover, no agent envies the donated bundle. \citet{BCFF21} improved this number to $n-2$ goods with nice cancelable valuations, and \citet{M24} extended this to monotone valuations.
    \citet{BCFF21} further showed that an $\efx$ allocation exists for $4$ agents with at most one donated good.   
    The number of donated goods was subsequently improved at the expense of achieving $(1-\varepsilon)$-$\efx$, instead of exact $\efx$ \citep{CGMMM21, ACGMM23, BBK22, CSJS23}.  

    Another pairwise fair notion is $\pmms$, introduced in \citep{CKMPSW19}, which is strictly stronger than $\efx$, but recently, \citet{BMP25} showed that
    $\pmms$ does not always exist even for three agents, where $\efx$ allocations have been shown to exist.

    Another line of research in fair allocation focuses on indivisible chores, where each agent is associated with a cost function. For additive cost functions, \cite{ZW24} obtained an $O(n^2)$-approximation for $\efx$, which was later improved to $4$-$\efx$ for any number of agents by \cite{GMQ25}. Independently, \cite{CS24} and \cite{afshinmehr2024approximateefxexacttefx} showed the existence of $2$-$\efx$ allocations for three agents, and recently \cite{garg2025existence2efxallocationschores} extended this to $2$-$\efx$ allocations for any number of agents. Exact $\efx$ allocations are known only in very restricted settings, and surprisingly, \cite{CS24} demonstrated that for monotone cost functions, there exist instances with no $\efx$ allocation. Another major open question—whether allocations can be both $\ef1$ and $\po$—was recently resolved affirmatively by \cite{mahara2025existencefairefficientallocation}.

    \section{Preliminaries}\label{prelims}
    An instance of discrete fair division is a tuple $\langle N,M,\{v_i\}_{i \in [n]}\rangle$, where
    $N=[n]=\{1,2,\ldots,n\}$ is a set of agents, $M$ is a set of $m$ indivisible goods, and $\{v_i\}_{i \in [n]}$ is a 
    profile of valuation functions, where $v_i: 2^M \to \mathbb{R}_{\geq 0}$ for each agent $i\in N$ determines $i$'s value for each subset of goods.     For notational simplicity, for a subset of goods $S$ and good $g \in S$, we use $S\setminus g$ to denote $S \setminus \{g\}$ and $g$ to denote \{g\}. 
    
    We consider valuation functions that are monotone, i.e., for any $S\subseteq T\subseteq M$, $v(S) \leq v(T)$. The valuation function $v(.)$ is called additive if, for any subset of items $S$, we have $v(S) = \Sigma_{g \in S} v(g)$. The valuation function $v(.)$ is called cancelable if, for any $S, T \subseteq M$ and any $g \in M \setminus (S \cup T)$, if $v(S \cup g) > v(T \cup g)$, then $v(S) > v(T)$, i.e., removing the same good from two different bundles would not change the relative preference between the two. Note that the class of cancelable valuation functions is a strict superclass of additive valuations.

    \paragraph{Multi-graph Instances.} 
        A fair division instance $\mathcal{I} = \langle [n], [m], \{v_i\}_{i \in [n]}\rangle$ on a \emph{multi-graph}\footnote{A multi-graph may contain multiple edges between two vertices.} is represented by a multi-graph $G = (V, E)$, where the $n$ agents correspond to the vertices in $V$, and the $m$ goods correspond to the edges in $E$. The structure is such that for every agent $i \in [n]$ and every subset of goods $S \subseteq [m]$, 
        \[
            v_i(S) = v_i(S \cap E(i)),
        \]
        where $E(i) \subseteq E$ is the set of goods incident to $i$, and $E(i,j)$ the set of edges between $i$ and $j$. For a multi-graph $G = (V, E)$, we define its skeleton as a graph $G' = (V, E')$  where $G'$ has the same set of vertices, and there is a single edge between two vertices if they share at least one edge in $G$, i.e., $i$ is connected to $j$ in $G'$ if $E(i, j) \neq \emptyset$ in $G$. We call a multi-graph \emph{triangle-free} if its skeleton does not contain any cycle of length three, i.e., the girth\footnote{The girth of a graph is the length of its shortest cycle.} of its skeleton is greater than or equal to four. In this work, we use the words \emph{good}, \emph{item}, and \emph{edge} interchangeably.

    \paragraph{Allocations and Orientations.} A partial allocation $\X = \langle X_1, X_2, \ldots, X_n \rangle$ is an ordered tuple of disjoint subsets of $[m]$, i.e, for every pair of distinct agents $i$ and $j$, we have $X_i, X_j \subseteq [m]$ and $X_i \cap X_j = \emptyset$. Here, $X_i$ denotes the bundle allocated to agent $i \in [n]$ in $X$. 
    A partial orientation is a partial allocation where $\forall i \in [n]$, we have $X_i\subseteq E(i)$.
    We say an allocation $\X$ is complete if $\bigcup\limits_{i \in [n]} X_i = [m]$.
 
    \paragraph{Envy, Strong Envy, $\efx$-Feasibility, $\efx$ Allocation, $\efx$ Cut, and $\efx$ Orientation.} 
    We say an agent $i$ envies a bundle $T$ with respect to bundle $S$ if $v_i(T) > v_i(S)$.
    Moreover, agent $i$ \emph{strongly envies} $T$ with respect to $S$
    if there exists a good $g \in T$ such that $v_i(S)<v_i(T \setminus g)$.
    Given an allocation $\X = \langle X_1, X_2, \ldots, X_n \rangle$, we say agent $i$ envies (respectively, strongly envies) agent $j$ 
    if $i$ envies (strongly envies) bundle $X_j$ with respect to bundle $X_i$.
    We say an allocation is $\efx$ if there is no strong envy between any pair of agents.
    Given a partition $\Y = \langle Y_1, \ldots, Y_k \rangle$ of a subset of goods into $k$ bundles, 
    we say bundle $Y_\ell$ is $\efx$-feasible for agent $i$ in $\Y$ if agent $i$ does not strongly envy any of the bundles in $\Y$ with respect to $Y_\ell$. We say the partition $\C = (C_1, C_2)$ of items in set $S$ is an $\efx$ cut of set $S$ for agent $i$ if both $C_1$ and $C_2$ are $\efx$-feasible for agent $i$ in $\C$. For a fair division instance on a multi-graph, the allocation $\X$ is an $\efx$ Orientation if it is an $\efx$ allocation and $X_i \subseteq E(i)$ for every $i \in [n]$. 

    Next, we demonstrate a property of $\efx$ orientations on multi-graphs via the following observation.

    \begin{observation}
    \label{one envy}
        For a multi-graph instance, consider a partial $\efx$ orientation $\X$ where an agent $i$ is envied by one of her neighbors $j$. Then, we must have $X_i \subseteq E(i,j)$. In particular, any vertex is envied by at most one agent in any $\efx$ orientation.
    \end{observation}

\section{EFX Allocations on Triangle-Free Multi-Graphs}

    In this section, we prove the main result of our work (\cref{thm:main_result}), the existence of an $\efx$ allocation for triangle-free multi-graphs involving an arbitrary number of goods and agents with monotone valuations.
    
\begin{theorem}\label{thm:main_result}
    For every muti-graph instance, $\efx$ allocations always exist when the skeleton of the underlying multi-graph contains no cycle of length three, and we can compute one in pseudo-polynomial time when the valuation functions are monotone and in polynomial time when the valuation functions are cancelable. 
\end{theorem}
    
     Our proof consists of three main phases: in the first two phases, we deal with partial orientations. We will define a set of properties in each phase and try to satisfy all of them. Note that in the second phase, we aim to add properties to the output orientation of the first phase while preserving the properties defined in the first phase. Finally, in the third phase, we allocate the remaining items to a third party, changing our orientation to an allocation, and complete our $\efx$ allocation. It is notable that the first two phases can be applied to any multi-graph instance; i.e., we use the triangle-free condition only in the final phase.

     We will now introduce each phase and prove our main result in the following subsections.

\subsection{Phase One}

    We generalize the concept of cut configurations, introduced by \citet{ADKMR25}, to the setting of all multi-graphs. Accordingly, their definition relied strongly on the structure of a bipartite multi-graph for which they showed the existence of an $\efx$ allocation. We define the concept of configuration for every pair of vertices in any multi-graph as follows: 

\begin{definition} [Cut Configuration]
    For every pair of agents $i$ and $j$, we fix two partitions of the set $E(i, j)$. The partition $(C_{i, j}^{(1)}, C_{i, j}^{(2)})$ of $E(i, j)$ is an $\efx$ cut for agent $i$ and is called the \emph{$i$-cut} configuration between agents $i$ and $j$. Similarly, the \emph{$j$-cut} configuration is the partition $(C_{j, i}^{(1)}, C_{j, i}^{(2)})$ of $E(i, j)$ that is an $\efx$ cut for agent $j$. Note that the first index indicates which agent makes the cut.

\end{definition}
     We aim to fix a specific configuration between every pair of vertices in our proof. We consider a sequence of agents denoted by $\sigma = [\sigma(1), \sigma(2), \cdots, \sigma(n)]$, where $\sigma(i) \in [n]$. Note that the sequence $\sigma$ is a permutation of agents in $[n]$. Every such sequence $\sigma$ imposes a fixed set of configurations on our input instance. More precisely, if $i$ precedes $j$ in $\sigma$ (denoted by $i \prec_\sigma j$), then we use the $j$-cut configuration for the set $E(i, j)$ in our construction. Formally, we define \emph{unit bundles} as follows:

 \begin{definition}[Unit Bundles]
    For every sequence $\sigma$ of agents, if $i \prec_\sigma j$, we define the bundles $C_{j,i}^{(1)}$ and $C_{j,i}^{(2)}$ of the \emph{$j$-cut} configuration between agents $i$ and $j$ as the \emph{unit bundles} between these two agents. We say a unit bundle $C$ is incident for agent $i$, if $C$ is a unit bundle of $E(i,j)$ for some agent $j$.
 \end{definition}

Thus, we always deal with a sequence that has fixed configurations in our graph. In this phase, we aim to find an appropriate sequence that satisfies certain properties. We use the concept of \emph{available} edges for an agent, introduced by \citet{ADKMR25}, during our proof several times.

\begin{definition}[$A_{i, j}(\X, \sigma)$ and $A_{j, i}(\X, \sigma)$]
    For a partial allocation $\X = \langle X_1, \cdots, X_n\rangle$ 
    and a sequence $\sigma$, and any two agents $i$ and $j$ where $i \prec_\sigma j$ , we define the sets $A_{i, j}(\X, \sigma) \subseteq E(i, j)$ and $A_{j, i}(\X, \sigma) \subseteq E(i, j)$ in four cases using Table \ref{tab:Ai,j definition}\footnote{Any other case will not happen during our algorithm. Therefore, these sets are not defined in other cases.}. We call $A_{i, j}(\X, \sigma)$ the set of \emph{available} edges in $E(i,j)$ for agent $i$, and similarly, $A_{j, i}(\X, \sigma)$ the set of \emph{available} edges in $E(i,j)$ for agent $j$.  
\end{definition}

    In our construction, for every pair of agents $i$ and $j$, we allocate at most one of the unit bundles in $E(i, j)$ to agent $i$. Accordingly, $\aij$ is in fact the most valuable unallocated unit bundle that agent $i$ can get from the set $E(i, j)$ subject to the mentioned restriction. Moreover, in our construction, if an agent like $i$ receives one item from a unit bundle, she receives the entire bundle, i.e., for an agent $i$ and unit bundle $C$, if there exists a good $g \in X_i \cap C$, then $C \subseteq X_i$.

    \begin{table}
    \centering
    \begin{tabular}{|c|c|c|}
    \hline
    & $A_{i,j}(\X, \sigma)$ & $A_{j,i}(\X, \sigma)$\\
    \hline
    $E(i,j) \cap X_k = \emptyset \ \text{for all} \ k \in [n]$ & $\argmax\limits_{C \in \{C_{j, i}^{(1)}, C_{j, i}^{(2)}\}} \{v_i(C)\}$ & $\argmax\limits_{C \in \{C_{j, i}^{(1)}, C_{j, i}^{(2)}\}} \{v_j(C)\}$ \\
    \hline
    $X_i \cap E(i,j) = \emptyset, X_j \cap E(i,j) \neq \emptyset$  & $E(i, j) \setminus X_j$ & $\emptyset$ \\
    \hline
    $ X_i \cap E(i,j) \neq \emptyset, X_j \cap E(i,j) = \emptyset$ & $\emptyset$ & $E(i, j) \setminus X_i$\\
    \hline
    $ X_i \cap E(i,j) \ne \emptyset, X_j \cap E(i,j) \ne \emptyset$ & $\emptyset$ & $\emptyset$\\
    \hline
    \end{tabular}
    \caption{Definition of $A_{i,j}(\X, \sigma)$ and $A_{j,i}(\X, \sigma)$ when $i \prec_\sigma j$ and the \emph{$j$-cut} configuration is fixed between agents $i$ and $j$.}\label{tab:Ai,j definition}
    \end{table}
    
Next, we introduce the key properties that we aim to satisfy in the first phase of our algorithm.

\paragraph{Key Properties:}
    In this phase, we search for a partial allocation $\X = \langle X_1, \cdots, X_n\rangle$ and a sequence $\sigma = [\sigma(1), \sigma(2), \cdots, \sigma(n)]$
    satisfying the following properties:

\begin{enumerate}
    \item \label{p1} $\X$ is an $\efx$ orientation.

    \item \label{p2} For any two agents $i$ and $j$ where $i \prec_\sigma j$, the items in $E(i, j)$ must be allocated according to the \emph{$j$-cut} configuration $(C_{j, i}^{(1)}, C_{j, i}^{(2)})$ to either one of their endpoints. Formally, one of the following must hold in $\X$: either no item in $E(i,j)$ is allocated or one of the unit bundles in the $j$-cut configuration is allocated to either $i$ and $j$ and the other unit bundle is unallocated or both unit bundles in the $j$-cut configuration are allocated, one to $i$ and one to $j$. Note that additional items may be allocated to $i$ or $j$.

    \item \label{p3} For any agent $i$ and any unallocated unit bundle $C$, we have $v_i(X_i)\geq v_i(C)$. 
    %For any pair of agents $i$ and $j$, we have $v_i(X_i) \ge v_i(A_{i, j}(\X, \sigma))$ and $v_i(X_i) \ge v_i(A_{j, i}(\X, \sigma))$
    %\ArA{Suggestion: For any pair of agents $i$ and $j$, and every unallocated unit bundle $C$ in $E(i,j)$, we have $v_i(X_i)\geq v_i(C).$} \KM{I like this suggestion.}

    \item \label{p4} The length of any envy path\footnote{$i_0 \rightarrow i_1 \rightarrow \cdots \rightarrow i_l$ is an envy path if each agent $i_q$ envies the agent $i_{q + 1}$.} in the graph is at most one. 
\end{enumerate}

The first three conditions are easy to ensure. Take an arbitrary $\sigma$ and fix the configuration between agents $i$ and $j$ to \emph{$j$-cut} whenever $i$ precedes $j$. Let the agents pick their favorite available unit bundle in order. Consider an agent $i$ and assume that it picks a unit bundle in $E(i,j)$. If $X_j \cap E(i,j) = \emptyset$, agent $i$ does not envy agent $j$. Otherwise, if agent $i$ precedes agent $j$ in sequence $\sigma$, the \emph{$j$-cut} configuration is used for $E(i,j)$ and agent $j$ is happy with both unit bundles of $E(i, j)$. If agent $j$ precedes agent $i$, the \emph{$i$-cut} configuration is used and agent $j$ picks before agent $i$. Thus, she will pick the unit bundle that is more valuable to her. Therefore, the only tricky part of phase one is to satisfy property \ref{p4} while maintaining the first three properties. 

We now prove that there exists an allocation $\X$ and sequence $\sigma$ satisfying all four properties simultaneously.

\subsubsection{Satisfying Properties (1)-(4)}

    We design an iterative algorithm to compute a partial allocation $\X$ and sequence $\sigma$ simultaneously such that they satisfy properties \ref{p1}-\ref{p4}. 
    
    In each step of our algorithm, we use two disjoint partial sequences $\sigma_L$ and $\sigma_R$ together with a set $U$ containing agents that are not in any of these sequences. Furthermore, $\X = \langle X_1, \cdots, X_n\rangle$ denotes the partial allocation in each step. More precisely, we will gradually build the sequence $\sigma$ from both sides by appending agents to the end of the sequence $\sigma_L$ or to the beginning of the sequence $\sigma_R$. This means that in each step, our final sequence $\sigma$ has the form of $\sigma = [\sigma_L(1), \cdots, \sigma_L(|\sigma_L|), u_1, \cdots, u_{|U|}, \sigma_R(1), \cdots, \sigma_R(|\sigma_R|)]$ where $u_i \in U$. Basically, the exact place of agents outside of $U$ is determined, and we only do not know the ordering of agents in $U$ in our final sequence. 
    Note that the configuration between any agent $i \in [n] \setminus U$ and any other agent $j$ is fixed in any step of the algorithm. 
    Also, for an agent $i\in U$, the bundle $X_i$ is not yet determined. In each step, we will gradually build $\sigma$ by fixing the position of some agents in $U$ in our sequence and allocating a unit bundle to each of them, and then removing them from $U$. 

    We begin with $\sigma_L = \sigma_R = []$ and $U = [n]$. In each step, we will decrease $|U|$ by at least one, while maintaining the following invariants:

\begin{enumerate}[(i)]
    \item \label{i1} Every agent in $\sigma_L$ and $\sigma_R$ has received a unit bundle that is incident to her and another agent in $\sigma_L$ or $\sigma_R$, that is, no agent in these two partial sequences has received any item incident to agents in $U$. 
    % Agents in $\sigma_L$ and $\sigma_R$ have received a set of items from each other, that is, no agent in these two partial sequences has received items valuable to agents in $U$.
    % \ArA{Every agent in $\sigma_L$ and $\sigma_R$ has received a unit bundle from an agent in $\sigma_L$ and $\sigma_R$, that is , no agent in these two partial sequences has received items valuable to agents in $U$.}
    % \MA{Every agent in $\sigma_L$ and $\sigma_R$ has received a unit bundle that is incident to her and another agent in $\sigma_L$ or $\sigma_R$, that is , no agent in these two partial sequences has received any item incident to agents in $U$.}\ArA{Let's use incident instead of valuable}

    \item \label{i2} For any agent $i \in U$, we have $X_i = \emptyset$.

    \item \label{i3} The allocation $\X$ is an orientation.

    \item \label{i4} Agents in $\sigma_L$ and $\sigma_R$ do not strongly envy any bundle in $\X$.

    \item \label{i5} Agents in $\sigma_L$ do not envy any agent in $\X$.

    \item \label{i6} Agents in $\sigma_R$ do not envy each other.

    \item \label{i7} For any agent $i\notin U$ and any other agent $j \in [n]$ property  \ref{p2} holds. Additionally, for any agent $i \notin U$, property \ref{p3} holds. 
    % Properties \ref{p2} and \ref{p3} are satisfied for any pair of vertices $i$ and $j$ that do not belong to $U$, i.e., $i, j \notin U$.
    % \ArA{$i,j \in \sigma_L \cup \sigma_R$.}
    % \ArA{For any agent $i\notin U$ and any other agent $j$ property  \ref{p2} holds, for property \ref{p3} holds for agent $i$.}
\end{enumerate}

    In the beginning, we have $\sigma_L = \sigma_R = []$ and $U = [n]$, which clearly satisfy all invariants. Thus, it remains to show that we can move in the space of partial allocations satisfying invariants \ref{i1}-\ref{i7} while decreasing the number of agents in $U$.

    Now, we go over a single step of our algorithm and show that we can decrease $|U|$ while maintaining all invariants. Consider a single step of our algorithm with $U \neq \emptyset$. Note that the configuration between any agent in $U$ and any agent in $[n] \setminus U$ is fixed before this step. Let $i$ be an arbitrary agent in $U$. We will append $i$ to the beginning of $\sigma_R$, and remove $i$ from $U$. Therefore, for any agent $u \in U$, the configuration between $u$ and $i$ will be fixed to \emph{$i$-cut}. Thereby, we have all the configurations in which agent $i$ is involved in. Although $\sigma$ is not complete yet, with a slight abuse of notation, we use the definition of $A_{i, j} (\X, \sigma)$. This is indeed valid, because the unit bundles incident to agent $i$ are fixed. Let $j = \argmax_{j \in [n]} v_i(A_{i, j} (\X, \sigma))$. If $j \notin U$, we allocate $A_{i, j} (\X, \sigma)$ to agent $i$ and stop. Otherwise, we add $j$ to the end of $\sigma_L$ and remove $j$ from $U$, thereby fixing configurations involving $j$. Let $k = \argmax_{k \in [n]} v_j(A_{j, k} (\X, \sigma))$. There are three cases:

\begin{itemize}
    \item \textbf{Case 1: $k = i$.} We allocate $A_{j, i} (\X, \sigma)$ to agent $j$ and then repeat the procedure for agent $i$, that is we will again find an agent $j = \argmax_{j \in [n]} v_i(A_{i, j} (\X, \sigma))$ and continue.

    \item \textbf{Case 2: $k \neq i$ and $k \notin U$.} We allocate $A_{j, k} (\X, \sigma)$ to agent $j$ and $A_{i, j} (\X, \sigma)$ to agent $i$ and stop.

    \item \textbf{Case 3: $k \neq i$ and $k \in U$.} We allocate $A_{j, k} (\X, \sigma)$ to agent $j$ and $A_{i, j} (\X, \sigma)$ to agent $i$. Then, we add $k$ to the beginning of $\sigma_R$ and then repeat the procedure that we did for agent $i$ exactly for agent $k$, that is we will again find an agent $j = \argmax_{j \in [n]} v_k(A_{k, j} (\X, \sigma))$ and continue.
\end{itemize}

    Formally, the procedure is defined using Algorithm \ref{alg:sequence}. Lemma \ref{lem:sequence} proves that this procedure maintains our seven invariants and decreases $|U|$:

\begin{lemma}\label{lem:sequence}
    \cref{alg:sequence} maintains invariants \ref{i1}-\ref{i7} and decreases $|U|$.
\end{lemma}

\begin{proof}

    First, note that the procedure defined by \cref{alg:sequence} terminates since it decreases $|U|$ by at least one. By construction, whenever an agent receives a unit bundle, it has already been removed from $U$ and been added to either $\sigma_L$ or $\sigma_R$. Moreover, whenever an agent receives a unit bundle incident to an agent in $U$, this agent will be removed from $U$ and will be added to either $\sigma_L$ or $\sigma_R$. Also, clearly, agents only receive incident unit bundles. Therefore, invariants \ref{i1}, \ref{i2}, and \ref{i3} hold. Every agent $i \notin U$ receives only a single unit bundle. Accordingly, whenever such an agent receives a unit bundle, it is chosen greedily among all of her unallocated incident unit bundles, meaning that invariant \ref{i7} is maintained. Therefore, we need to prove that invariants \ref{i4}-\ref{i6} hold.

    Consider agents $i$ and $j$ as defined in the beginning of the \cref{alg:sequence}. Note that if $j \notin U$, only agent $i$ is added to $\sigma_R$ and is receiving a unit bundle from $E(i, j)$. By invariant \ref{i7}, agents previously in $\sigma_L$ or $\sigma_R$ do not envy agent $i$. Moreover, by invariant \ref{i1}, agent $i$ will not envy any agent previously in $\sigma_L$ or $\sigma_R$, as they do not have a unit bundle valuable to agent $i$. Thus, invariants \ref{i4}-\ref{i6} are maintained. 

    Note that during \cref{alg:sequence}, agents that are added to $\sigma_R$ are denoted by $i$, and agents that are added to $\sigma_L$ are denoted by  $j$.
    %We show that whenever an agent is added to $\sigma$ or ..., after the execution of while loop is over invariants hold.

    First, we show that \cref{alg:sequence} maintains invariant \ref{i5}.  
    Whenever an agent $j$ is added to the end of $\sigma_L$, she greedily picks a unit bundle. Note that whenever agent $j$ receives a unit bundle, the other agents in $\sigma_L$ or $\sigma_R$ have not yet been assigned a unit bundle incident to agent $j$ by invariant \ref{i1}, meaning that agent $j$ does not envy any other agent. 
    By invariant \ref{i7}, agents previously in $\sigma_L$ do not envy the unit bundle that $j$ is receiving. 
    Therefore, invariant \ref{i5} holds.

    Next, we show that \cref{alg:sequence} maintains invariant \ref{i6}.  
    Whenever an agent $i$ is added to $\sigma_R$ and eventually receives a unit bundle, she does not envy agents who were previously added to $\sigma_R$ since they are not receiving a unit bundle incident to agent $i$ by invariant \ref{i1}. Moreover, agents previously in $\sigma_R$ do not envy agent $i$ by invariant \ref{i7}. 
    Therefore, invariant \ref{i6} also holds. 

    Finally, we show that \cref{alg:sequence} maintains invariant \ref{i4}.   
    Note that by invariant \ref{i5}, we only need to show that agents in $\sigma_R$ do not strongly envy any bundle in $\X$. Consider an agent $j$ in $\sigma_L$ and an agent $i$ in $\sigma_R$. If $X_j \cap E(i, j) = \emptyset$, then, agent $i$ does not envy agent $j$. Otherwise, agent $j$ has received a unit bundle from $E(i, j)$ and the configuration between agents $i$ and $j$ is fixed to \emph{$i$-cut}. Therefore, whenever agent $i$ receives a unit bundle, $\aij$ is the unit bundle in $E(i, j)$ that is not allocated to $j$, and since agent $i$ picks an available unit bundle incident to her greedily, she receives a unit bundle with a value at least as much as $\aij$ to her. 
    Thus, by the definition of our configurations, agent $i$ will not strongly envy agent $j$, meaning that invariant \ref{i4} also holds.   
\end{proof}

    % We show that invariants \ref{i4}-\ref{i6} are maintained in every iteration of our algorithm, using the induction hypothesis. Note that these invariants were satisfied before the first iteration of our algorithm, meaning that the base case of our induction holds. Now, consider a single iteration of \cref{alg:sequence} with agents $i$ and $j$ as defined, and assume invariants \ref{i4}-\ref{i6} were satisfied before this iteration. Note that if $j \notin U$, only agent $i$ is added to $\sigma_R$ and is receiving a unit bundle from $E(i, j)$. By invariant \ref{i7}, agents previously in $\sigma_L$ or $\sigma_R$ do not envy agent $i$. Moreover, by invariant \ref{i1}, agent $i$ will not envy any agent previously in $\sigma_L$ or $\sigma_R$, as they do not have a unit bundle valuable to agent $i$. Thus, invariants \ref{i4}-\ref{i6} are maintained. Now, assume $j \in U$. By construction, $j$ will be removed from $U$ and will be added to the end of $\sigma_L$. First, note that in all the cases, whenever agent $j$ receives a unit bundle, other agents previously in $\sigma_L$ or $\sigma_R$ have not yet assigned a unit bundle incident to agent $j$, meaning that agent $j$ will not envy any other agent. By invariant \ref{i7}, we have that agents previously in $\sigma_L$ do not envy the unit bundle that $j$ is receiving. Therefore, as agent $j$ is also being added to $\sigma_L$ and does not envy any agent, invariant \ref{i5} remains satisfied. 

\begin{algorithm}[h]
\caption{Sequence Augmentation}\label{alg:sequence}
\KwIn{Partial allocation $\X$, sequences $\sigma_L$ and $\sigma_R$, and set $U$ satisfying invariants (i)-(vii).}
\KwOut{Partial allocation $\X$, sequences $\sigma_L$ and $\sigma_R$, and set $U$ satisfying invariants (i)-(vii) with strictly decreased $|U|$.}

$i \gets \text{ an arbitrary vertex in $U$}$ 

$\sigma_R 
\gets [i, \sigma_R (1), \cdots, \sigma_R(|\sigma_R|)]$

$U \gets U \setminus \{i\}$

$j \gets \argmax_{j \in [n]} v_i(A_{i, j} (\X, \sigma))$

\While{$j \in U$}{

    $\sigma_L \gets [\sigma_L (1), \cdots, \sigma_L(|\sigma_L|), j]$
    
    $U \gets U \setminus \{j\}$

    $k \gets \argmax_{k \in [n]} v_j(A_{j, k} (\X, \sigma))$

    \If{$k = i$}{
    
        $X_j \gets A_{j, i} (\X, \sigma)$

        $j \gets \argmax_{j \in [n]} v_i(A_{i, j} (\X, \sigma))$
        
    }
    \ElseIf{$k \notin U$}{

        $X_j \gets A_{j, k} (\X, \sigma)$

        \Break   
    
    }
    \ElseIf{$k \in U$}{

        $X_j \gets A_{j, k} (\X, \sigma)$

        $X_i \gets A_{i, j} (\X, \sigma)$

        $i \gets k$   

        $\sigma_R \gets [i, \sigma_R (1), \cdots, \sigma_R(|\sigma_R|)]$

        $U \gets U \setminus \{i\}$

        $j \gets \argmax_{j \in [n]} v_i(A_{i, j} (\X, \sigma))$

        %\ArA{Here, we mean to start over but this time do not use an arbitrary $i$, but instead use $k$ as $i$, but this is not exactly correct at least in the pseudo code}
        %\ArA{I believe in order to get a pseudo code in correspondence with your description of the algorithm, you may need something like "go to line 2" here. We can use other things like while but I think that makes it messy}
        
        %\ArA{Another way to fix this is to remove the first line in this algorithm that chooses $i$ arbitrarily, and give $i$ as an input to the algorithm.
         %Although this is not entirely the same as your description of the algorithm}
    
    }

}

$X_i \gets A_{i, j} (\X, \sigma)$

\Return $\X, \sigma_L, \sigma_R, U$

\end{algorithm}

Now, we are able to prove the main Lemma of this section:

\begin{lemma}
    There exists a complete sequence $\sigma$ and a partial allocation $\X$ satisfying properties \ref{p1}-\ref{p4}.
\end{lemma}

\begin{proof}
    We begin with $\sigma_L = \sigma_R = []$ and $U = [n]$. While $U \neq \emptyset$, we decrease $|U|$ and maintain invariants \ref{i1}-\ref{i7} using  \cref{alg:sequence}. Finally, we achieve a complete sequence $\sigma$ and an allocation satisfying invariants \ref{i1}-\ref{i7}. We will prove that the output allocation $\X$ of this procedure satisfies properties \ref{p1}-\ref{p4}. By invariant \ref{i7}, it is clear that properties \ref{p2} and \ref{p3} are satisfied. Furthermore, invariant \ref{i4} proves that the allocation is $\efx$. Combining with invariant \ref{i3}, we have satisfied property \ref{p1}. Thus, it only remains to show that property \ref{p4} holds. By invariants \ref{i5} and \ref{i6}, we have that the only envy between agents can be from an agent in $\sigma_R$ to an agent in $\sigma_L$. This means that the length of any envy path is at most one, thereby satisfying property \ref{p4}.
\end{proof}

Now that we have proved the existence of an allocation $\X$ and sequence $\sigma$ satisfying properties \ref{p1}-\ref{p4}, we have completed our first phase. Before moving to the second phase, we give an intuitive observation on the sequence $\sigma$:

\begin{observation}

    Let $\sigma$ be the output sequence for phase one. We construct an allocation $\X'$ as follows: we fix the configurations imposed by $\sigma$, i.e., for any pair of agents $i$ and $j$ where $i \prec_\sigma j$, we use the \emph{$j$-cut} configuration between agents $i$ and $j$. Then, we run a greedy algorithm using $\sigma$ as a picking sequence. More precisely, the greedy algorithm consists of $n$ iterations, where in each iteration $k$, agent $\sigma(k)$ will greedily pick the most preferred unallocated unit bundle for her. The output of such a greedy algorithm using $\sigma$ as a picking sequence is identical to the allocation $\X$, which was the output of phase one, i.e., $\X' = \X$.
    
\end{observation}

\subsection{Phase Two}
    Let the allocation $\X$ and sequence $\sigma$ be the outputs of the previous phase. The sequence $\sigma$ will be fixed for the rest of our algorithm. We aim to add three more properties to the allocation $\X$. We begin by defining some useful notation and then introduce three new properties. Since $\sigma$ is fixed, we will drop it from our notation, for example, we write $A_{ij}(\X)$ instead of $A_{ij}(\X,\sigma)$.

\subsubsection{Some Useful Notations}

    For a partial allocation $\X$ 
    and a sequence $\sigma$ satisfying properties \ref{p1}-\ref{p4}, we introduce the following definitions:

\begin{itemize}

    \item Given any pair of agents $i,j$, we define $\bijo$, $\bijt$, $\bjio$, and $\bjit$\  as follows:

    \begin{itemize}

        \item     If both of the unit bundles in $E(i,j)$ are unallocated, then $\bijo = \bjit$, and they are one of the unit bundles in $E(i,j)$, and $\bijt=\bjio$ are the other unit bundle in $E(i,j)$.
    
        \item     If agent $i$ possess a unit bundle in $E(i,j)$, then $\bijo=\bijt=\emptyset$, and the same follows for agent $j$.

        \item    If agent $i$ does not possess a unit bundle in $E(i,j)$, and $j$ possess a unit bundle in $E(i,j)$, then $\bijo=\bijt$ and they are the unallocated unit bundle in $E(i,j)$,
        and the same follows for agent $j$.
    \end{itemize}

    \item For $i \in [n]$, let $\bio= \bigcup_{j \in [n]\setminus \{i\}} \bijo$ and 
    $\bit= \bigcup_{j \in [n]\setminus \{i\}} \bijt$.

    \item For $i \in [n]$, $\ui$ is the set of all unallocated edges incident to $i$. Note that $\bio, \bit \subseteq U_i(\X)$. 
\end{itemize}

Next, we introduce our new properties that we wish to add to our partial allocation:

\paragraph{Key Properties:} We search for a partial allocation $\X = \langle X_1, \cdots, X_n\rangle$ that satisfies the following properties in addition to properties \ref{p1}-\ref{p4}:

\begin{enumerate}[start = 5]

    \item \label{p5} For any non-envied agent $i \in [n]$, we have $\bio = \emptyset$.

    \item \label{p6} For any non-envied agent $i \in [n]$, we have $v_i(X_i) \ge v_i(\ui)$
    
    \item \label{p7} For any envied agent $i \in [n]$, if $j$ envies $i$, we have:  
                    $$v_i(X_i) \geq v_i(X_j \cup \bio) \text{ and } v_i(X_i) \geq v_i(X_j \cup \bit).$$ 
\end{enumerate}

    We will now introduce an algorithm that satisfies these properties by allocating or unallocating some unit bundles to agents.  The procedure is formally defined in Algorithm \ref{alg:prop567}. 
    The algorithm, first checks if there exists an agent $i$ for whom property \ref{p5} does not hold, i.e., $\bio \ne \emptyset$, 
    then it allocates $X_i \cup \bio$ to agent $i$.
    Whenever that property \ref{p5} holds, \cref{alg:prop567} checks if there exists an agent $i$ for whom property \ref{p6} does not hold, i.e., $v_i(X_i)<v_i(\ui)$. 
    Then, for every $j$ which there exists an unallocated unit bundle in $E(i,j)$, agent $i$ swaps it  with her unit bundle in $E(i,j)$, leaving the unit bundle from every other $j$
    untouched. Note that the previous $\ui$ is now a subset of the new $X_i$.
    Whenever that properties \ref{p5} and \ref{p6} hold, \cref{alg:prop567} checks if there exists an agent $i$ for whom property \ref{p7} does not hold, 
    so if agent $j$ envies $i$, then $v_i(X_i)< v_i(X_j \cup \biu)$ for some $u\in \{1,2\}$. 
    Then, agents $i$ and $j$ swap the unit bundles in $E(i,j)$ with each other, and then $X_i \cup \biu$ is allocated to agent $i$.

    After the termination of our algorithm, by its definition, properties \ref{p5}-\ref{p7} are satisfied. Therefore, all we need to show is that our algorithm terminates in polynomial time while maintaining the first four properties. We will show that each if statement preserves our first four properties.

\begin{algorithm}
\caption{Satisfying (5)-(7)}\label{alg:prop567}
\KwIn{Sequence $\sigma$ and allocation $\X$ with properties  \ref{p1}-\ref{p4}}
\KwOut{Sequence $\sigma$ and allocation $\X$ with properties \ref{p1}-\ref{p7}}

\While{at least one of properties (5)-(7) is not satisfied}{
\If{there exists a non-envied agent $i\in [n]$ such that $\bio \neq \emptyset$}{
    $X_i\gets X_i \cup \bio$  \label{non-envied-shares} \\
}
\ElseIf{there exists a non-envied agent $i\in [n]$ such that $v_i(X_i)<v_i(\ui)$}{
    let $T = \{j: E(i,j)\cap \ui=\emptyset\}$\\
    $X_i\gets (X_i\cap \bigcup_{j\in T}E(i,j)) \cup \ui$\\
}
\ElseIf{there exists an agent $i$ envied by $j$ such that $v_i(X_i)< v_i(X_j \cup \biu)$ for some $u\in \{1,2\}$}{
    Let $(C_1,C_2)$ be the unit bundles of $E(i, j)$.\\
    Swap unit bundles $C_1$ and $C_2$ between agents $i$ and $j$. \label{swap}\\
    $X_i\gets X_i \cup \biu$\label{add after swap}\\
}
}

\Return $\X, \sigma$
\end{algorithm}

    We now introduce a potential function to show that our algorithm terminates after a polynomial number of iterations. Let $\phi_1(\X)$ be the number of envied agents, $\phi_2(\X)$ be the number of non-envied agents violating property \ref{p6}, and $\phi_3(\X)$ be the number of non-envied agents breaking property \ref{p5}. We define the lexicographic potential function $\phi(\X)$ to be 
    $$ \phi(\X) = \langle \phi_1(\X),\phi_2(\X),\phi_3(\X) \rangle ,$$‌ 
    which we will show that it decreases in lexicographic order in each iteration.

\begin{observation}
\label{b remains ub}
    If the allocation $\X$ and sequence $\sigma$ satisfy property \ref{p2}, then
    for all agents $i$ and $j$, we have that $\bijo, \bijt \in \{C^{(1)},C^{(2)}, \emptyset\}$, where $C^{(1)},C^{(2)}$ are the unit bundles between $i$ and $j$.  
    Moreover, if \xs\ satisfy property \ref{p3}, too, then for any pair of  
    agents $i$ and $j$, we have $v_i(X_i) \ge v_i(\bijo)$ and $v_i(X_i) \ge v_i(\bijt)$. 
\end{observation}

\begin{lemma} \label{preserve 1-4}
    If $(\X,\sigma)$ satisfies properties \ref{p1}-\ref{p4}, and for some unenvied agent $i$, 
    we take some unit bundles from agent $i$, and allocate some other unallocated incident unit bundles to $i$
    such that agent $i$ receives at most one unit bundle from any agent $j$, and $v_i(X_i)$ does not decrease, then $(\X,\sigma)$ satisfies properties \ref{p1}-\ref{p4}. Furthermore, no new envy from an agent to another will be created.
\end{lemma}
\begin{proof}
    Denote the bundle of agent $i$ before the change by $X_i$, and denote her new bundle by $X'_i$.
    Note that the agent $i$ is the only agent whose bundle changes, and she only receives her incident unit bundles;
    therefore, the allocation remains an orientation.
    Next, we show that agent $i$ remains non-envied.
    Since allocation remains an orientation, we get that for every agent $j$,
    $X'_i\cap E(j) = X'_i \cap E(i,j)$, so $X'_i\cap E(j)$ is a unit bundle $C$.
    Either $C$ was unallocated or $C\subseteq X_i$.
    Agent $j$ does not envy $C = X'_i\cap E(j)$, in the first case,  by property \ref{p3}, and in the second case, 
    by the fact that $X_i$ was non-envied.

    Moreover, for every pair of agents $j$ and $p$ such that $j\ne i$, since bundle $X_j$ does not change, 
    and no ones value for her bundle decreases, if agent $p$ did not envy (respectively, strongly envy) agent $j$,
    agent $p$ still does not envy (strongly envy) agent $j$.
    Hence, since allocation was $\efx$, it remains $\efx$ and property \ref{p1} still holds.
    Furthermore, since no new envy is created, and the length of any envy path was at most one, this property still holds after the update; therefore, property \ref{p4} still holds.

    In addition, since agent $i$ is the only agent whose bundle changes, and she only gets incident unit bundles with a restriction that
    she does not receives two unit bundles from an incident agent, property \ref{p2} still holds.
    Additionally, since the new unallocated unit bundles are either were unallocated in the previous allocation or were allocated to agent $i$ who was nonenvied, property \ref{p3} still holds.
\end{proof}

\begin{lemma}
\label{lem ai}
    During the execution of Algorithm \ref{alg:prop567}, whenever the first if-statement is executed, the potential function \pot strictly decreases while properties \ref{p1}-\ref{p4} remain satisfied.
\end{lemma}
\begin{proof}
    Given that we have entered the first if statement, $\X$ does not satisfy property \ref{p5}, and agent $i$ is a non-envied agent violating property \ref{p5}.

   % \textcolor{ForestGreen}{
    By definition, $\bio =\cup_{j \ne i} \bijo$, and $\bijo$ is a unit bundle in $E(i,j)$, 
    and also if $X_i \cap E(i,j)\ne \emptyset$, then $\bijo = \emptyset$.
    Hence, by setting $X_i \gets X_i \cup \bio$, the conditions of \cref{preserve 1-4} satisfy. Therefore, properties \ref{p1}-\ref{p4}
    still hold.

    Next, we show that \pot\ decreases.
    Since, by \cref{preserve 1-4}, no new envy is created, the number of envied agents, \poto, does not increase.
    If \poto\ decreases, \pot\ decreases, so assume not. Therefore,  the set of non-envied agents is the same as before the execution of 
    line \ref{non-envied-shares}. Also, since $\bio$ is a union of unallocated unit bundles, which will be  allocated to agent $i$,
    for every agent $j$, the set $\uj$ would be a subset of itself after the execution of line \ref{non-envied-shares}.
    Consequently, every non-envied agent who satisfies Property \ref{p6} will continue to satisfy it after the execution of line \ref{non-envied-shares}.
    Therefore, \pott\ does not increase.
    Also, note that if an arbitrary agent $j$ did not violate property \ref{p5} before the execution, she still does not violate it, while $i$ no longer violates property \ref{p5}, so $\phi_3(\X)$ strictly decreases. Therefore, $\phi(\X)$ strictly decreases after the execution.
\end{proof}

\begin{observation}\label{obs:structure}
    If allocation $\X$ and $\sigma$ satisfy properties \ref{p1}-\ref{p5} then looking at the allocated edges we find
    \begin{itemize}
        \item Between two non-envied agents $i$ and $j$, both unit bundles must be allocated by property \ref{p5}.
        
        \item Between envied agent $i$ and agent $j$ who is the sole agent who envies $i$, both bundles must also be allocated. 
        This is because $i$ must possess some unit bundle that $j$ envies, and due to property \ref{p4}, $j$ must be non-envied, 
        so she has the remaining unit bundle by property \ref{p5}.
        
        \item Between envied agent $i$ and non-envied agent $j$ who does not envy $i$, there must be a single unallocated unit bundle. 
        Agent $i$ cannot possess any unit bundle from $E(i,j)$; otherwise, she would be strongly envied. Also, $j$ must have a single unit bundle from $E(i,j)$, by property \ref{p5}.
        
        \item Between two envied agents $i$ and $j$, both unit bundles must be
        unallocated.
    \end{itemize}
\end{observation}

\begin{lemma}
During the execution of Algorithm \ref{alg:prop567}, whenever the second if-statement is executed, the potential function \pot\ strictly decreases while properties \ref{p1}-\ref{p4} remain satisfied.
\end{lemma}

\begin{proof}
    Given that we have entered the second if statement, $\X$ must satisfy properties \ref{p1}-\ref{p5} but does not satisfy property \ref{p6}, and $i$ is a non-envied agent violating the mentioned property.

    By Observation \ref{obs:structure}, $\ui$ contains only items between agent $i$ and envied agents that are not envied by $i$, and for every such envied agent $j$, 
    a unit bundle of $E(i, j)$ has been allocated to agent $i$, and the other unit bundle is unallocated. Inside the second if statement, we basically swap these two unit bundles, i.e., allocate the previously unallocated one to $i$ and release the other unit bundle. 
    %\textcolor{ForestGreen}{
    Hence, the conditions of \cref{preserve 1-4} satisfy. Therefore, properties \ref{p1}-\ref{p4}
    still hold.
    %}

    % \textcolor{red}{
    % Since property \ref{p3} hold before the execution and since agent $i$ was non-envied, no one envies agent $i$ after the execution.
    % In addition, since agent $i$ was not envied, the unit bundles she drops are not envied by anyone, so property \ref{p3} holds after the execution.
    % Properties \ref{p1} and \ref{p2} are clearly still held. 
    % Moreover, property \ref{p4} still holds as no new envy is created.
    % }

    Next, we show that \pot\ decreases.
    Since, by \cref{preserve 1-4}, no new envy is created, the number of envied agents, \poto, does not increase.
    If \poto\ decreases, \pot\ decreases, so assume not. Therefore,  the set of non-envied agents is the same as before. 
    For any non-envied agent $j\neq i$, by Observation \ref{obs:structure}, it is clear that $\uj$ does not change, 
    meaning that if she used to satisfy property \ref{p6}, she still does. 
    Also, agent $i$ now satisfies property \ref{p6} as the new $\ui$ is a subset of what she held before. 
    This means that $\phi_2(\X)$ must strictly decrease, resulting in the decrease of $\phi(X)$.
\end{proof}

\begin{lemma}
During the execution of Algorithm \ref{alg:prop567}, whenever the third if-statement is executed, the potential function \pot\  strictly decreases while properties \ref{p1}-\ref{p4} remain satisfied.
\end{lemma}
\begin{proof}
    Given that we have entered the third if statement, $\X$ must satisfy properties \ref{p1}-\ref{p6}, but does not satisfy property \ref{p7}, and $i$ is an envied agent violating the mentioned property.

    %\textcolor{ForestGreen}{
    First, we show that agent $i$ becomes non-envied.
    Let $C\subseteq E(i,j)$ be the unit bundle that agent $i$ receives after the swap.
    Since agent $j$ envied the bundle held by agent $i$, after swapping the unit bundles of $E(i,j)$, agent $j$ no longer envies agent $i$. 
    Also, by \cref{obs:structure}, $\biu\cap E(j)=\emptyset$, meaning that agent $j$ does not envy agent $i$ after the execution of line \ref{add after swap}.
    Additionally, for every agent $p\notin \{i,j\}$, $\biu \cap E(p)$ is an unallocated unit bundle, so by property \ref{p3}, 
    $v_p(X_p)\geq  v_p(\biu \cap E(p))= v_p( (C\cup \biu) \cap E(p))= v_p(C \cup \biu)$. 
    Therefore, agent $i$ becomes non-envied.
    %}

 %   \textcolor{ForestGreen}{
    In addition, agent $j$ remains non-envied because the only part of $X_j$ that is changed is the unit bundle in $E(i,j)$, and agent $i$ does not envy this unit bundle by the fact that before the swap we had:
    $$v_i(X_i) < v_i(X_j \cup \biu) = v_i((X_j \cap E(i,j)) \cup \biu).$$
    %}

%\textcolor{ForestGreen}{
    Moreover, for every pair of agents $q\notin \{i,j\}$ and $p$, since bundle $X_q$ does not change, 
    and no ones value for her bundle decreases, if agent $p$ did not envy (respectively, strongly envy) agent $q$,
    agent $p$ still does not envy (strongly envy) agent $q$.
    Hence, since allocation was $\efx$, it remains $\efx$ and property \ref{p1} still holds.
    Furthermore, since no new envy is created, and the length of any envy path was at most one, this property still holds after the update; therefore, property \ref{p4} still holds.
 %   }

%\textcolor{ForestGreen}{
    In addition, since agents $i$ and $j$ are the only agents whose bundles change, and they only receive incident unit bundles with a restriction that they do not receive two unit bundles from an incident agent, property \ref{p2} still holds.
    Additionally, since the new unallocated unit bundles are a subset of previous unallocated unit bundles, property \ref{p3} still holds.
 %   }

%\textcolor{ForestGreen}{
    Since every non-envied agent remained non-envied, and envied agent $i$ became non-envied,
    $\phi_1(\X)$, and as a result $\phi(\X)$ strictly decreases.
    %}
%     \textcolor{red}{
%     Since agent $j$ envied the bundle held by agent $i$, after swapping the unit bundles of $E(i,j)$, she no longer envies agent $i$. 
%     Also, by \cref{obs:structure}, $\biu\cap E(j)=\emptyset$, meaning that agent $j$ does not envy agent $i$ after the execution of line \ref{add after swap}.
%     Additionally, for every agent $p\notin \{i,j\}$, $\biu \cap E(i,p)$ is an unallocated unit bundle, and by property \ref{p3}, 
%     $v_p(X_p)\geq v_p(\biu \cap E(i,p))= v_p(\biu \cap E(p))= v_p(\biu)$. 
%     Also, since right after the swap in line \ref{swap}, $X_i$ is a unit bundle $C$ in $E(i,j)$, we have $X_i\cap E(p)=\emptyset$.
%     Hence, after the execution of line \ref{add after swap}, agent $p$ does not envy $X_i=C\cup \biu$.
%     Therefore, agent $i$ must become non-envied.
%     In addition, agent $j$ remains non-envied because the only part of $X_j$ that is changed is the unit bundle in $E(i,j)$, and agent $i$ does not envy that by the fact before the swap we had:
%     $$v_i(X_i) < v_i(X_j \cup \biu)=v_i((X_j \cap E(i,j)) \cup \biu).$$
% }
    % \textcolor{red}{
    % Moreover, no agent will have a bundle with a lower value from her perspective after the execution of line \ref{add after swap}.
    % Therefore, the other non-envied agents remain non-envied as well.
    % Thus, $\phi_1(\X)$, and as a result $\phi(\X)$ strictly decreases.
    % Additionally, properties \ref{p1}-\ref{p4} are still satisfied by similar arguments as the proof of \cref{lem ai}.}
\end{proof}

Given that $\phi_1(\X)$, $\phi_2(\X)$ and $\phi_3(\X)$ are each at most $n$, Algorithm \ref{alg:prop567} satisfies properties \ref{p1}-\ref{p7} in at most $n^3$ iterations.

\subsection{Phase Three}
    Let the allocation $\X$ and sequence $\sigma$ be the outputs of the first two phases. After satisfying all our desired properties, 
    we now allocate the remaining unallocated unit bundles to a third-party to obtain a full $\efx$ allocation. Note that this is the only phase where we use the assumption that our input instance is triangle-free.
    
%     Note that by \cref{obs:structure} the only unallocated unit bundles reside in the following cases:

% \begin{itemize}
%     \item \textbf{Case 1: }Between an envied agent $i$ and a non-envied agent $j$ where agent $j$ has received a unit bundle from $E(i, j)$ but agent $i$ has not received such a bundle. Moreover, agent $j$ does not envy agent $i$.

%     \item \textbf{Case 2: }Between two envied agents $i$ and $k$. Let $j$ and $r$ be the agents envying $i$ and $k$, respectively. Note that since our input instance does not contain a cycle of length three, we have $j \neq r$.
% \end{itemize}

    Algorithm \ref{alg:dumping} takes $(\X,\sigma)$ that satisfy properties \ref{p1}-\ref{p7} as input, and then for every agent $j$ who envies some agent $i$, 
    adds $\bio$ to agent $j$. We prove that this algorithm outputs a complete $\efx$ allocation in \cref{full efx}.

    %Simply saying, if we allocate all the available bundles for envied vertices while maintaining the $\efx$ property, we will prove Theorem \ref{thm:main_result}. Algorithm \ref{alg:dumping} dumps the remaining agents to a third-party and ouputs our desired $\efx$ allocation.

\begin{algorithm}
\caption{Dumping Remaining Items}\label{alg:dumping}
\KwIn{Sequence $\sigma$ and allocation $\X$ satisfying properties \ref{p1}-\ref{p7}}
\KwOut{An $\efx$ allocation}

% \While{there exists an envied agent $i$ with $\bio \neq \emptyset$}{

%     Let $j$ be the agent who envies $i$.

%     $X_j \gets X_j \cup B_i{(1)}(\X, \sigma)$

% }

\For{every envied agent $i$}{
    Let $j$ be the agent who envies $i$. \hfill \tcp{There is only one by \cref{one envy}.}

    $X_j \gets X_j \cup \bio$

}
\Return $\X$

\end{algorithm}

\begin{lemma}
\label{full efx}
    Algorithm \ref{alg:dumping} terminates in at most $n$ rounds, and outputs an $\efx$ allocation.
\end{lemma}

\begin{proof}
    Algorithm \ref{alg:dumping} terminates at most n iterations since there are fewer than $n$ envied agents. 

    Prior to the execution of Algorithm \ref{alg:dumping}, \cref{obs:structure} implies that every unallocated unit bundle resides in one of the following cases:

\begin{itemize}
    \item \textbf{Case 1:} Between an envied agent $i$ and a non-envied agent $j$ where agent $j$ does not envy $i$. Here, agent $j$ has received a unit bundle from $E(i, j)$ but agent $i$ has not received a unit bundle from $E(i,j)$.
    % Between an envied agent $i$ and a non-envied agent $j$ where agent $j$ has received a unit bundle from $E(i, j)$ but agent $i$ has not received such a bundle. Moreover, agent $j$ does not envy agent $i$.
    % \PM{I think it is better to say "Between an envied agent $i$ and a non-envied agent $j$ where agent $j$ does not envy $i$. Here, agent $j$ has received a unit bundle from $E(i, j)$ but agent $i$ has not received such a bundle." at the start. it is in the definition of i and j, not just a fact}

    \item \textbf{Case 2: }Between two envied agents $i$ and $k$. Let $j$ and $l$ be the agents envying $i$ and $k$, respectively. Note that since our input instance does not contain a cycle of length three, we have $j \neq l$.
\end{itemize}

    We first argue that after execution of Algorithm \ref{alg:dumping}, all of the items will be allocated, which is to show that all the unallocated unit bundles will be allocated.
    In case 1, $\bio$ contains the unallocated unit bundle in $E(i,j)$, and it will be dumped to the agent who envies agent $i$.
    In case 2, since both of the unit bundles of $E(i,k)$ are unallocated, by definition, we get that $\biko = E(i,k) \setminus \bkio$. 
    Therefore, since $\biko \subseteq \bio$, and $\bkio \subseteq \bko$, unit bundles of $E(i,k)$ will be allocated to $j$ and $l$.
    
    Next, we show that allocation remains $\efx$.
    Suppose agent $i$ is an arbitrary non-envied, and she envies agents $j_1,\ldots,j_z$ prior to execution of algorithm \ref{alg:dumping}.
    Since, allocation was $\efx$, and no agent loses any item, we only need to show that no agent envies agent $i$ after that 
    $X'_i  = X_i \cup \bo{j_1} \cup \ldots \cup \bo{j_z}$ is allocated to $i$.

    For any agent $j_p \in \{j_1,\ldots, j_z\}$, we have that $\bigcup_{q\in [z]\setminus \{p\}} \bo{j_q} \cap E(j_p) = \emptyset$
    because the underlying multi-graph is triangle-free, and $j_p$ and $j_q$ are both adjacent to $i$, so $j_p$ and $j_q$ are not adjacent.
    Hence, $$v_{j_p}(X_{j_p}) \geq v_{j_p}(X_i \cup \bo{j_p} )) = v_{j_p}(X'_i),$$
    where the first inequality is due to property \ref{p7}. 
    
    % Next, assume $l \notin \{i, j_1, \ldots, j_z\}$. Note that since the underlying multi-graph is triangle-free, 
    % we have either $X_i \cap E(l) =  \emptyset$ or $\bigcup_{q\in [z]} \bo{j_q} \cap E(l)=\emptyset$. Thus, we only need to show that
    % $v_r(X_r) \geq v_r(X_i)$ and $v_r(X_r) \geq v_r(\bigcup_{q\in [z]} \bo{j_q})$.
    % The first inequality holds since agent $i$ was non-envied. Moreover, $\bigcup_{q\in [z]} \bo{j_q} \cap E(l) = \bigcup_{q\in [z]} B_{j_q,l}^{(1)}(\X,\sigma)$.
    % If $l$ is non-envied, then we get
    % $$ v_r(X_r) \geq v_r(\ui) \geq v_r(\bigcup_{q\in [z]} B_{j_q,l}^{(1)}(\X,\sigma)),$$
    % where the first inequality comes from property \ref{p6}, and the last inequality comes from the fact $\bigcup_{q\in [z]} B_{j_q,l}^{(1)}(\X,\sigma) \subseteq \ui$,
    % and the fact that valuation functions are monotone.

    % If $l$ is envied, then   $B_{j_q,l}^{(1)}(\X,\sigma) = B_{l,j_q}^{(2)}(\X,\sigma)$, so
    % $$\bigcup_{q\in [z]} \bo{j_q} \cap E(l) = \bigcup_{q\in [z]} B_{j_q,l}^{(1)}(\X,\sigma) = \bigcup_{q\in [z]} B_{l,j_q}^{(2)}(\X,\sigma) \subseteq 
    % B_{l}^{(2)}(\X,\sigma),$$

    % Thus, $$v_r(X_r) \geq v_r(B_{l}^{(2)}(\X,\sigma)) \geq v_r(\bigcup_{q\in [z]} \bo{j_q} \cap E(l)) = v_r(\bigcup_{q\in [z]} \bo{j_q}),$$
    % where the first inequality is due to property \ref{p7}, and the last inequality is based on the previous argument and the fact valuation functions are monotone.
    % Therefore, the proof is complete.
    
    Now, it remain to show the same for agents $r \notin \{i, j_1, \ldots, j_z\}$. Note that since the underlying multi-graph is triangle-free, 
    $r$ is either not adjacent to $i$, meaning $X_i \cap E(r) =  \emptyset$ or $r$ is not adjacent to every $j_p\in\{ j_1, \ldots, j_z\}$, 
    meaning $\bigcup_{q\in [z]} \bo{j_q} \cap E(r)=\emptyset$. Thus, we only need to show that
    $v_r(X_r) \geq v_r(X_i)$ and $v_r(X_r) \geq v_r(\bigcup_{q\in [z]} \bo{j_q})$.
    The first inequality holds since agent $i$ was non-envied.
    For the second inequality, note that
    $$\bigcup_{q\in [z]} \bo{j_q} \cap E(r) = \bigcup_{q\in [z]} \boxy{j_q}{r}.$$
    
    If $r$ is non-envied, since every unit bundle $\boxy{j_q}{r}$ is in $U_r(\X)$ then by monotonicity and property \ref{p6}:
    $$ v_r(X_r) \geq v_r(U_r(\X) \geq v_r(\bigcup_{q\in [z]} \boxy{j_q}{r})= v_r(\bigcup_{q\in [z]} \bo{j_q}).$$
    
    If $r$ is envied, since every $j_q$ is also envied, both unit bundles between $r$ and $j_q$ are unallocated, so $\boxy{j_q}{r} = \btxy{r}{j_q}$, and therefore, 
    $$\bigcup_{q\in [z]} \bo{j_q} \cap E(r) = \bigcup_{q\in [z]} \boxy{j_q}{r} = \bigcup_{q\in [z]} \btxy{r}{j_q} \subseteq \bt{r},$$
    
    Thus, recalling the monotonicity of the valuation functions and by property \ref{p7}, we find
    $$v_r(X_r) \geq v_r(\bt{r}) \geq v_r(\bigcup_{q\in [z]} \bo{j_q} \cap E(r))= v_r(\bigcup_{q\in [z]} \bo{j_q}).$$
    Either way, the desired inequality holds. Therefore, the proof is complete.
\end{proof}
    % Assume we have agents $i$ and $j$ as described in the first case of our remaining items as described above. Also, let $k$ be the agent envying $i$. Algorithm \ref{alg:dumping} allocates the remaining items between $i$ and $j$ to agent $k$. Since the allocation was an orientation previously and our instance was triangle-free, there are no items between agents $j$ and $k$, and thus, we have that $v_j(X_k) = 0$. Therefore, after allocating these items to agent $k$, agent $j$ will not envy her by property \ref{p6}. Moreover, agent $i$ will not envy agent $k$ by property \ref{p7} and the allocation remains $\efx$.

    % Now, assume we have agents $i$ and $k$ as described in the second case. Let $k$ and $r$ be the agents who envy agents $i$ and $j$, respectively. Since the allocation was an orientation previously and our instance was triangle-free, there are no items between agents $j$ and $k$ and also between $i$ and $r$, and thus, we have that $v_k(X_j) = 0$ and $v_i(X_r) = 0$. Algorithm \ref{alg:dumping} allocates one of the unit bundles between agents $i$ and $k$ to agent $j$ and the other one to agent $r$. The allocation indeed remains $\efx$ by property \ref{p7}.

Now, we have proven our main result, i.e., we can obtain an $\efx$ allocation for triangle-free multi-graphs if we do the following: first, execute \Cref{alg:sequence} for at most $n$ times to satisfy the first four properties. Then, execute Algorithms \ref{alg:prop567}, and \ref{alg:dumping} in the mentioned order.

\subsection{Running Time}

    In this part, we analyze the running time of our algorithm. By Lemma \ref{lem:sequence}, executing Algorithm \ref{alg:sequence} at most $n$ times satisfies properties \ref{p1}-\ref{p4}, where each execution takes polynomial time. It is also clear that Algorithms \ref{alg:prop567} and \ref{alg:dumping} terminate in polynomial time, as discussed in previous sections. Furthermore, $A_i(\X, \sigma)$ can also be computed in polynomial time. Therefore, the only part of our algorithm that determines the running time of it is the computation of the \emph{Cut Configurations} between any pair of adjacent vertices. Therefore, we briefly discuss the running time of computing such configurations for different classes of valuation functions.

    \paragraph{$\efx$ Cut Between Two Agents.} 
    It is well known that the simple algorithm introduced by \cite{PR20}, later referred to as the PR algorithm in \cite{ACGMM23}, 
    takes as input a monotone valuation function $v$, a set of goods $S$, and a natural number $k$, and outputs a partition
    $(Y_1,\ldots, Y_k)$ of $S$ such that every bundle in the partition is $\efx$-feasible for an agent with valuation $v$.
    The PR algorithm runs in pseudo-polynomial time for monotone valuations (see the proof of Theorem 3.49 in \cite{AGS25}). 
    Furthermore, \cite{AGS25} proposed a modified version of the PR algorithm and showed that, for $k=2$ and a cancelable valuation function, the modified algorithm runs in polynomial time (see Lemma 4.7 in \cite{AGS25}).

    Combining the results above from \cite{AGS25}, we can conclude that our result terminates in polynomial time for cancelable valuation functions, which are a strict generalization of additive valuation, and in pseudo-polynomial time for general monotone valuations.

\section{Conclusion}

In this work, we have studied a model that captures a setting where each item is relevant to at most two agents, while any pair of agents can have an arbitrary number of items that they both value, represented via a multi-graph. We advanced the understanding of fair allocation by resolving a key open question on the existence of $\efx$ allocations in multi-graph valuation settings. By proving existence whenever the underlying graph contains no three-cycles, our work strictly generalizes prior results and significantly broadens the class of instances for which $\efx$ fairness can be guaranteed. Beyond existence, we contribute algorithmic insights by establishing a pseudo-polynomial procedure for computing $\efx$ allocations under monotone valuations, which becomes polynomial when valuations are cancelable. These results mark one of the few known cases in which $\efx$ allocations exist for an arbitrary number of agents, thereby moving the boundary of tractable and guaranteed fairness further than previously known. Future directions include tightening the structural assumptions on the graph, that is, allowing cycles of length three to be presented in the graph structure, and proving the existence of $\efx$ allocations. We conjecture that any fair division instance represented via a multi-graph admits an $\efx$ allocation. Moreover, computing such an allocation might be possible by adding certain useful properties to our initial orientation (the output of our first phase). Another interesting direction is to explore trade-offs between fairness and efficiency, for example, by approximating Social Welfare or Nash Social Welfare.

Finally, we believe that the algorithmic insights gained by our work on triangle-free multi-graphs will contribute to future research on the existence and computation of $\efx$ allocations in more general settings.

\bibliographystyle{plainnat}
\bibliography{mybibliography}

\end{document}